\documentclass[norunningheads]{llncs}
\usepackage[usenames, dvipsnames]{color}
\usepackage{fullpage}
\usepackage{graphicx}
\usepackage{epsfig,epsf}
\usepackage{epstopdf}
\usepackage{graphics}
\usepackage{array}
\usepackage{amsmath}
\usepackage{amssymb}
\usepackage{comment}
\usepackage{algorithm}
\usepackage{algorithmic}

\newtheorem{cl}{Claim}
\date{}

\title{Hamiltonicity in Convex Bipartite Graphs}
\author{Kowsika P \and Divya V  \and Sadagopan N} 
\institute{Department of Computer Science and Engineering,\\ Indian Institute of Information Technology Design and Manufacturing, Kancheepuram, India. \\
\email{\{coe15b008,ced15i024,sadagopan\}@iiitdm.ac.in}}
\begin{document}
\maketitle
\thispagestyle{plain}
\titlerunning{Hamiltonicity in Convex Bipartite Graphs}
\begin{abstract}
For a connected graph, the Hamiltonian  cycle (path) is a simple cycle (path) that spans all the vertices in the graph.   It is known from \cite{muller,garey} that HAMILTONIAN CYCLE (PATH) are NP-complete in general graphs and chordal bipartite graphs.  A convex bipartite graph $G$ with bipartition $(X,Y)$ and an ordering $X=(x_1,\ldots,x_n)$, is a bipartite graph such that for each $y \in Y$, the neighborhood of $y$ in $X$ appears consecutively.  $G$ is said to have convexity with respect to $X$.  Further, convex bipartite graphs are a subclass of chordal bipartite graphs. In this paper, we present a necessary and sufficient condition for the existence of a Hamiltonian cycle in convex bipartite graphs and further we obtain a linear-time algorithm for this graph class.  We also show that Chvatal's necessary condition is sufficient for convex bipartite graphs.  The closely related problem is HAMILTONIAN PATH whose complexity is open in convex bipartite graphs.  We classify the class of convex bipartite graphs as {\em monotone} and {\em non-monotone} graphs.  For monotone convex bipartite graphs, we present a linear-time algorithm to output a Hamiltonian path.  We believe that these results can be used to obtain algorithms for Hamiltonian path problem in non-monotone convex bipartite graphs.  It is important to highlight (a) in \cite{keil,esha}, it is incorrectly claimed that Hamiltonian path problem in convex bipartite graphs is polynomial-time solvable by referring to \cite{muller} which actually discusses Hamiltonian cycle  (b) the algorithm appeared in \cite{esha} for the longest path problem (Hamiltonian path problem) in biconvex and convex bipartite graphs have an error and it does not compute an optimum solution always.  We present an infinite set of counterexamples in support of our claim.  
\footnotetext{This work is partially supported by DST-ECRA project - ECR/2017/001442}
\end{abstract}
\section{Introduction}
Classical problems such as HAMILTONIAN PATH (CYCLE), LONGEST PATH have applications in the field of mathematics and computing, for example, \cite{garey}.  Given a connected graph $G$, the Hamiltonian path (cycle) problem asks for a simple spanning path (cycle) in $G$. 

While there has been a constant study on the structural front to obtain a necessary and sufficient condition for the existence of {a} Hamiltonian path (cycle), on the complexity front, it is known that these problems are NP-complete \cite{garey}.  We do not know of any necessary and sufficient condition for the existence of  {a} Hamiltonian path (cycle) in general graphs, however, we know of necessary conditions, and sufficient conditions \cite{west}.

To understand the complexity of NP-complete problems, it is natural to either restrict the input based on some structural observations or to focus on special graphs.  Focusing on special graphs such as chordal graphs, chordal bipartite graphs, etc., is a promising direction as problems such as VERTEX COVER, CLIQUE, COLORING have polynomial-time algorithms on chordal and chordal bipartite graphs.  As far as {the} Hamiltonian path and cycle problems are concerned, both structural and algorithmic study even in popular special graphs {is} considered to be a challenging task.  This is due to, (i) there are no necessary and sufficient conditions for the existence of Hamiltonian path in special graphs such as chordal and chordal bipartite which are the most sought after graph classes in the literature, (ii) these problems are NP-complete even in chordal and chordal bipartite graphs \cite{muller}.

When a problem is NP-complete in chordal bipartite graphs, it is natural to restrict the input further and study the complexity status on convex bipartite graphs which {are} a well-known subclass of chordal bipartite graphs.  Due to the convexity property, it is easy to see that any cycle of length at least 6 has a chord in it and hence, convex bipartite graphs are a subclass of chordal bipartite graphs.  Similarly, interval graphs which {are} a subclass of chordal {graphs} is a candidate graph class for problems that are NP-complete in chordal graphs.   For example, HAMILTONIAN CYCLE and VERTEX COVER are polynomial-time solvable in {the} interval and convex bipartite graphs. Interestingly, convex bipartite graphs and interval graphs have a structural relationship and hence there is a polynomial-time reduction that converts an instance of interval graph to an instance of convex bipartite graph, and vice-versa.  Further, for many problems, algorithms designed for interval graphs can be used to solve problems in convex bipartite graphs.  One such problem is HAMILTONIAN CYCLE whose solution in convex bipartite graphs can be obtained from interval graphs \cite{keil}.  The algorithm presented in \cite{keil} requires $O(n+m)$ effort to output a Hamiltonian cycle in a convex bipartite {graph} if exists.

The study of Hamiltonicity includes HAMILTONIAN CYCLE and HAMILTONIAN PATH.  HAMILTONIAN CYCLE is polynomial-time solvable in convex bipartite graphs \cite{muller,keil} whereas, to the best of our knowledge, the complexity of HAMILTONIAN PATH is open.  The objective of this paper is {two-fold}; (a) we present a characterization for HAMILTONIAN CYCLE as there are no known structural characterizations for the existence of Hamiltonian cycles in convex bipartite graphs.   Further, we present a linear-time algorithm for this problem.  (b) we also present a linear-time algorithm for Hamiltonian path problem in monotone convex bipartite graphs.  We believe that the results presented can be extended for further research to study the complexity of HAMILTONIAN PATH in non-monotone convex bipartite graphs and LONGEST PATH and MIN-LEAF SPANNING TREE in convex bipartite graphs.   We wish to highlight that the algorithms appeared in \cite{esha} for longest path problems in biconvex and convex graphs (convex bipartite graphs) have an error and it does not compute an optimum solution always.  We present an infinite set of {counterexamples} in support of our claim.   All {of our counterexamples} are {\em non-monotone} convex bipartite graphs and for {\em monotone} convex bipartite graphs, the algorithm presented in \cite{esha} gives a {Hamiltonian} path by incurring $O(n^6)$ effort, while our algorithm is linear. \\\\
{\bf {Roadmap}:} We next present preliminaries required to present our results.  We shall present results on HAMILTONIAN CYCLE in Section \ref{hcycle}. In Section \ref{hpath}, we shall present our results on HAMILTONIAN PATH.  Results on non-monotone convex bipartite graphs {are} presented in Section \ref{nonmono}.
\section{Graph Preliminaries}
All graphs used in this paper are simple, connected and unweighted and we follow the standard definitions and notation \cite{west}.  For a graph $G$, denote the vertex set by $V(G)$ and the edge set by $E(G)=\{\{u,v\} ~|~u$ is adjacent to $v$ in $G$ $\}$.  The neighborhood of a vertex $v$ in $G$ is $N_G(v) = \{u ~|~ \{u, v\} \in E(G)\}$.  The degree of a vertex $v$ in $G$ is $d_G(v) = |N_G(v)|$.  We denote $\delta(G) = \min \{d_G(v) ~|~ v \in V(G)\}$.  A vertex $v$ is said to be {\em pendant} if $d_G(v)=1$. \\
A bipartite graph $G(X,Y)$ with partitions $X$ and $Y$ is a convex bipartite graph if there is an ordering of $X=(x_1,\ldots,x_n)$ such that for all $y \in Y$, $N_G(y)$ is consecutive with respect to the ordering of $X$, and $G$ is said to have convexity with respect to $X$.  Let $n=|X|$ and $m=|Y|$.  Define the set $X_{p..q} \subset X$ to be $X_{p..q}=\{x_p,\ldots,x_q\}, 1 \leq p < q \leq n$, the set $N_G[X_{p..q}] = \{y ~|~ N_{G}(y) \subseteq \{x_{p},\ldots,x_{q}\}\}$ and the set $N'_G[X_{p..q}] = \{y ~|~ y \in Y$ and $|N_{G}(y) \cap X_{p..q}| \geq 2\}$.   For a vertex $y \in Y$ with $N_G(y)=\{x_i,\ldots,x_j\},1 \leq i < j \leq n$, we define $left(y)=i$ and $right(y)=j$.  
The set $X_{p..q} \subset X$ is said to be {\em maximal} if $|N_G[X_{p..q}]| = q-p+1$ and there is no $X_{p'..q'}$ such that $|N_G[X_{p'..q'}]| = q'-p'+1$, { $1 < p' < p < q \le q'< n$ or $1 < p' \le p < q < q'< n$}, and we refer to this set as {\em maximal}$(X_{p..q})$.  Similarly, the set $X_{p..q}$ is said to be {\em minimal} if $|N_G[X_{p..q}]| = q-p+1$ and there is no $X_{p'..q'}$ such that $|N_G[X_{p'..q'}]| = q'-p'+1$, {$1 < p < p' < q' \leq q <n$ or $1 < p \leq p' < q' < q <n$}, and we refer to this set as {\em minimal}$(X_{p..q})$.  Throughout out this paper, $G$ denotes a convex bipartite graph with convexity on $X$. 
\section{Hamiltonian cycles in convex bipartite graphs}\label{hcycle}
In this section, we shall present a characterization for the existence of Hamiltonian cycles in convex bipartite graphs.   We shall define a property, namely {\tt Property A} in convex bipartite graphs.  Further, we present a transformation to transform an input convex bipartite graph into an interval graph.  Subsequently, {we show that Hamiltonian cycles in  convex bipartite graphs exist if and only if  Hamiltonian cycles in the interval graphs exist}.   Our structural results are new and a part of algorithmic results make use of the results presented in \cite{muller,keil}.  Towards the end, we shall visit the Hamiltonian path problem and present a solution for monotone convex bipartite graphs.
\begin{definition} Let $G$ be a convex bipartite graph. $G$ is said to satisfy {\tt Property A}, if it satisfies the following properties{, } \\
\begin{enumerate}
\item Upper bound: \\
{For any set} $X_{p..q}$,\\
$|N_G[X_{p..q}]| \leq q-p, \text{if $1=p<q<n$ or $1<p<q=n$ or $1<p<q<n$}$.\\
\item Lower bound: \\
{For any integer $j\in [1..n-1]$},\\
\begin{equation*}
	|\bigcup_{i=1}^j N'_G[X_{p_i..q_i}]| \geq (\sum_{i=1}^{j-1} q_i-p_i)+ q_j-p_j, 
	\begin{cases} & \text{if $j>1$, $ 1\leq p_i<q_i \leq p_{i+1}<q_{i+1}\leq n $}\\ & \text{if $j=1$, $1\leq p_j<q_j\leq n$}\end{cases}
\end{equation*}
\item {Cardinality bound}: \\ $|N_G[X_{1..n}]|=|N'_G[X_{1..n}]|=n$\\
\end{enumerate}
\end{definition}
{\bf Transformation: Convex bipartite graphs to Interval graphs} \\
Let $G$ be a convex bipartite graph with $|X|=|Y|$.  We shall now define {the corresponding interval graph} $I$ as follows; $V(I)= X \cup Y$ and $E(I)=E(G) \cup E'$, where $E'=\{ \{y_i,y_j\} ~|~ y_i,y_j \in Y$ and $N_G(y_i)\cap N_G(y_j) \neq \emptyset, 1 \leq i \not= j \leq |Y| \}$.  Interestingly, Muller \cite{muller} has related {the} Hamiltonian cycles in interval graphs and convex bipartite graphs and his result is given below. 
\begin{theorem} \cite{muller}
Let $G$ be a convex bipartite graph with $|X|=|Y|$ and $I$ be the corresponding interval graph of $G$. $G$ has a Hamiltonian cycle iff $I$ has a Hamiltonian cycle. 
\end{theorem}
It is a well-known result \cite{golu} that the maximal cliques of an interval graph $I$ can be linearly ordered, such that, for every vertex $x$ of $I$, the maximal cliques containing $x$ occurs consecutively in the ordering.  To present our results on HAMILTONIAN CYCLE in {the corresponding interval graphs}, we shall reuse the definitions and {notation} given in \cite{keil} by Keil. \\
A maximal clique in $I$ corresponds to a set of mutually overlapping intervals in $I$.  Every clique $C$ in $I$ has a {\em representative point} $r$ which is contained in exactly those intervals corresponding to vertices in $C$.  That the cliques of $I$ can be linearly ordered by their representative points follows from the following lemma from Gilmore and Hoffman \cite{gandh}.
\begin{lemma} \cite{gandh}
The maximal cliques of an interval graph $I$ can be linearly ordered, such that, for every vertex $x$ of $I$, the maximal cliques containing $x$ occurs consecutively.
\end{lemma} 
We call a clique of $I$, $C_i$, when its representative point $r_i$ is the $i^{th}$ smallest of the representative points of the cliques of $I$.
A vertex $v$ that appears in a clique $C_i$ is called a {\em conductor} for $C_i$ if $v$ also appears in $C_{i+1}$. Let $L(C_i)$ be the set $\{C_1,C_2,\ldots,C_i\}$ of the first $i$ cliques of $I$.  A path $P$ in $I$ is {\em spanning} for $L(C_i)$ if $P$ contains all vertices of $I$ that appear only in $L(C_i)$ and $P$ has two conductors of $C_i$ as endpoints.   Let $R_i$ be the set of representative points of the cliques containing vertex $v_i$.  A {\em point embedding} $Q$ of a path $P$, with $n$ vertices, is an assignment of a point $q(v_i)$ from $R_i$ to $v_i$ such that $q(v_i) \in R_{i+1}$ for $1 \leq i \leq n-1$. A path is {\em straight} if it has a point embedding $Q$ with the property that $q(v_i)\leq q(v_{i+1})$ for $1 \leq i \leq n-2$.  In a straight path $P$ with point embedding $Q$, $v_i$ is the active conductor between $q(v_{i-1})$ and $q(v_i)$, $2 \leq i \leq n-1$ and $v_1$ is the active conductor between the smallest point in $R_1$ and $q(v_1)$.  A path $P$, with endpoints $A$ and $B$, that spans $L(C_i)$ is said to be {\em U-shaped} if there exists a vertex $v$ in $P$ that appears only in $C_1$ such that the subpaths of $P$ from $v$ to $A$ and from $v$ to $B$ are both straight.  Such a vertex $v$ is called the base of the U-shaped path $P$. The below characterization due to Keil \cite{keil} is an important result and our result is based on this key result. 
\begin{lemma} \cite{keil}
Let $I$ be an interval graph with $m$ maximal cliques.  $I$ has a Hamiltonian cycle iff there exists a U-shaped spanning path for $L(C_i)$, $1\leq i \leq m-1$.
\end{lemma}
The interval graph that we will be working with is not an arbitrary interval graph, instead, it is the graph obtained from a convex bipartite graph through the above transformation.   We shall now present a series of structural results using which we establish that the transformed interval graph has indeed a U-shaped spanning path.  As a consequence, we obtain a Hamiltonian cycle in the interval graph.   Throughout this {section,} we work with a convex bipartite graph $G$ with partitions $X$ and $Y$ and $I$ be the corresponding interval graph of $G$. 
\begin{lemma}
\label{l1}
Let $G$ be a convex bipartite graph and $I$ be the corresponding interval graph of $G$.  Then, the number of maximal cliques in $I$ is $|X|$.
\end{lemma}
\begin{proof}
By our construction, for each $x \in X$, $N_G(x)$ induces a clique in $I$.  Further, $\{x\} \cup N_G(x)$ is also a clique in $I$.  Again, by the construction of $I$, no two elements of $X$ are adjacent in $I$.  Thus,  $\{x\} \cup N_G(x)$ is a maximal clique in $I$.  Therefore, the number of maximal cliques in $I$ is $|X|$. \qed
\end{proof}
\begin{corollary}
Each clique $C_i$ of $I$ contains exactly one vertex $x_i \in X$ such that, $x_i$ is the $i^{th}$ vertex in the convex ordering on $X$.  
\end{corollary} 
\begin{proof}
Follows from Lemma \ref{l1}. \qed
\end{proof}
\begin{algorithm}[H]
\textit{Input:} A convex bipartite graph $G$ satisfying Property A and the corresponding interval graph $I$ of $G$. \\ \textit{Output:} A Hamiltonian cycle in $I$. 
\begin{algorithmic}[1]
\STATE Generate the ordered list of maximal cliques in $I$.  For each clique $C_i$, let $i$ be the height of $C_i$, $1 \leq i \leq n$.  For each vertex $v$, find the level[$v$] which is the height of the highest clique that contains $v$.
\STATE {Initially,} all vertices of $I$ are unlabeled.  The algorithm proceeds by visiting the cliques in increasing order of their height. After a clique $C_i$ is visited, a U-shaped spanning path $P$ for $L(C_i)$ containing only labeled vertices is found.
\STATE The endpoints of $P$ are labeled $A$ and $B$ while the interior vertices in $P$ are labeled $O$. {Any vertex,} not in $P$ remains unlabeled. As the algorithm {proceeds,} $A$ and $B$ are updated and unlabeled vertices {in path P} are changed to $O$.  When a clique $C_i$ is visited, there are three possibilities;
\IF{$i=1$}
{
	\STATE Choose two conductors $y_1$ and $y_2$ of $C_1$ such that level[$y_1$] is minimum and level[$y_2$] is {the} second minimum.  Label $y_1$ with $A$ and $y_2$ with $B$. $P(C_1)=(A,x_1,B)$ and $x_1$ is labeled $O$ ($x_1$ is the base of the U-shaped path).  
}
\ELSIF{$1<i<n$}
{
	\STATE The spanning path $P(C_i)$ is found from $P(C_{i-1})$.  Consider the endpoints $A$ and $B$ of $P(C_{i-1})$.
	\IF{level[$A$]=$i$ or level[$B$]=$i$}
	{
		\STATE Suppose level[$A$]=$i$, then the path moves from $A$ through $x_i$ to any unlabeled conductor $y$ of $C_i$ such that level[$y$] is minimum.  $A$ is relabeled $O$,  $y$ is labeled $A$ and $x_i$ is labeled  $O$.  Label of $B$ is unchanged.

	}
	\ELSIF{level[$A$]$>i$ and level[$B$]$>i$}
	{
		\STATE Choose the minimum of level[$A$] and level[$B$], say $A$.   Then update the path as follows; the path $P(C_{i-1})$ from $A$ moves through $x_i$ to any unlabeled conductor $y$ of $C_i$ such that level[$y$] is minimum. $A$ is relabeled $O$,  $y$ is labeled $A$ and $x_i$ is labeled  $O$.  Label of $B$ is unchanged.
	}
	\ENDIF
}
\ELSIF{$i=n$}
{
	\STATE Join the endpoints $A$ and $B$ of $P(C_{n-1})$ with $x_n$.
}
\ENDIF
\end{algorithmic}
\caption{Hamiltonian cycles in the transformed interval graphs}\label{hcalgo}
\end{algorithm}
\begin{lemma} 
\label{l2}
Let $G$ be a convex bipartite graph and $I$ be the corresponding interval graph of $G$.  Then, $\bigcup_{i=1}^{n} V(C_i)=V(G)$, where $V(C_i)$ denotes the vertex set of $C_i$.
\end{lemma}
\begin{proof}
From Lemma \ref{l1}, we know that each $x \in X$ appears in some maximal clique in $I$.  Further, each $y \in Y$ is adjacent to some $x \in X$ and therefore $x$ and $y$ together appear in some maximal clique in $I$.  Thus, $\bigcup_{i=1}^{n} V(C_i)=V(G)$ follows. \qed
\end{proof}
\begin{lemma}
\label{x=y}
Let $G$ be a convex bipartite graph.  If $G$ satisfies Property A, then $|X|=|Y|$ and $\forall y \in Y$, $d_G(y) \geq 2$.
\end{lemma}
\begin{proof}
By the definition of Property A, $|N_G[X_{1..n}]|=|N'_G[X_{1..n}]|=n$. Observe that $Y=W \cup Z${, where} $W = \{y~|~y \in Y$ and $d_G(y)= 1\}$, $Z = \{y~|~y \in Y$ and $d_G(y) \geq 2\}$ and $W \cap Z = \emptyset$.  Note that by the definition, $N'_G[X_{1..n}]=Z$ and $|Z|=n$.   Similarly, $N_G[X_{1..n}]=Z \cup W$ and $|W|+|Z|=n$.  This implies that $|W|=0$.  Therefore, $Y=Z$.   Further, $|Z|=|Y|=n=|X|$. Hence, the claim follows. \qed
\end{proof}
\begin{corollary}
Let $G$ be a convex bipartite graph. If $G$ satisfies Property A, then $\forall u \in Y$, $u$ appears in at least two consecutive maximal cliques of $I$.
\end{corollary}
\begin{proof}
From Lemma \ref{x=y}, we know that for all $y \in Y$, $d_G(y) \geq 2$.  That is, each $y$ is adjacent to at least two vertices in $X$.  From Lemma \ref{l1}, we know that each vertex in $X$ is part of exactly one maximal clique.  Therefore, the claim follows. \qed
\end{proof}
Given a convex bipartite graph satisfying Property A, we shall next present an algorithm (Algorithm \ref{hcalgo}) which {outputs} a Hamiltonian cycle.  Later, we show that convex bipartite graphs satisfying Property A always possess Hamiltonian cycles.   Our algorithm is a modified version of the algorithm presented in \cite{keil} in the context of interval graphs obtained through the above transformation.  \\\\
\noindent
\textbf{Proof of correctness of Algorithm \ref{hcalgo}:} As part of correctness, we need to show that our algorithm always {outputs} a Hamiltonian cycle.  We first show that at each iteration, our algorithm has two conductors through which extension of {the} path between iterations is guaranteed.  Secondly, we show that at each iteration, we include exactly one $x \in X$ and one $y \in Y$ in the path.\\
\begin{theorem}
\label{hci}
Let $G$ be a convex bipartite graph satisfying Property A and $I$ be the corresponding interval graph of $G$.   Algorithm \ref{hcalgo} always {outputs} a Hamiltonian cycle in $I$. 
\end{theorem}
\begin{proof}
The correctness of the algorithm follows from the following claims.
\begin{cl} For every iteration $i$, $1\leq i\leq n-1$, the spanning path $P(C_i)$ obtained from the algorithm has two conductors labeled with $A$ and $B$ as endpoints.    
\end{cl}
\begin{proof} The proof is by contradiction.  Since $G$ is connected, we find at least one conductor between any two adjacent cliques.  Let $i$ be the least index such that $P(C_i)$ has exactly one conductor.  This implies that $P(C_{i-1})$ has two conductors labeled with $A$ and $B$ as endpoints.   At iteration $i$, without loss of generality, {assume that the extension of the path is from $A$ and $A$ is relabeled $O$. The} label of $B$ is unchanged and is still the conductor.  Suppose, between $C_i$ and $C_{i+1}$, we do not find a conductor to be labeled $A$.   That is, $P(C_i)$ has exactly one conductor labeled $B$.  Note that until $P(C_{i-1})$, the vertices $(x_1,\ldots,x_{i-1})$ are included in the path.  Further, $x_i$ can be reached from $A$ of $P(C_{i-1})$.  However, {a} further extension to reach $x_{i+1}$ is not possible as we have exactly one conductor, i.e., there is no $y \in Y$ and $y$ is {an unlabeled} conductor which is adjacent to both $x_i$ and $x_{i+1}$.  Due to the above structural observation, we observe that {$|N'_G[X_{1..i+1}]|=i$ and this implies that $|N_G[X_{i+1..n}]| = n-i > n-(i+1)$,} contradicting the upper bound of Property A.  Therefore, we get two conductors between iterations. \qed
\end{proof}
In what follows from the above claim is {that} we obtain a $U$-shaped path in $P(C_{n-1})$ with $x_1$ as the base vertex.  Further, using $x_n$, we obtain a Hamiltonian cycle in $I$.  We observe a stronger result that the Hamiltonian cycle alternates between the vertices of $X$ and $Y$ which we shall present next.
\begin{cl} Let $Z$ = $\{y~|~y \in Y, y$ is labeled as $O$ in the spanning path $P(C_i)$ $\}$.  The spanning path $P(C_i)$, $1 \leq i \leq n-1$ is such that $|Z|$ = $i-1$.  
\end{cl}
\begin{proof}  The proof is by contradiction. 
\textbf{Case 1:} $|Z| < i-1$.  Note that $|N'_G[X_{1..i}]|< i-1$ which contradicts the lower bound of Property A. 
\textbf{Case 2:} $|Z| > i-1$. Observe that there exists an iteration $j$ in which at least two vertices of $Y$ are labeled $O$ and choose the least $j$ to work with.  {Thus we get, $|N_G[X_{1..j}]|=j-1$ and $|N_G[X_{1..j+1}]| \geq (j-1)+2 > (j+1)-1$,} which contradicts the upper bound of Property A. \qed
\end{proof}
The above two claims ensure that the path $P(C_i)$ constructed alternates between a vertex of $X$ and $Y$ and using $x_n$ we obtain a Hamiltonian cycle in $I$. \qed
\end{proof}
\begin{theorem}
\label{hci2}
Let $G$ be a convex bipartite graph satisfying Property A.  $G$ always has a Hamiltonian cycle. 
\end{theorem}
\begin{proof}
We first obtain the corresponding interval graph $I$ of $G$ using the transformation explained in this section.   From Theorem \ref{hci}, we know that $I$ has a Hamiltonian cycle.  From the construction of Hamiltonian cycle $C$ in $I$, it is clear that the cycle $C$ does not contain edges $\{y_i,y_j\}$ such that $\{y_i,y_j\} \in E'$.   Thus, $C$ is also a Hamiltonian cycle in $G$. \qed
\end{proof}
{\bf Run-time Analysis:} Given a convex bipartite graph $G$ satisfying Property A, if $n=|X|$, then $n-1 \leq |Y| \leq n+1$.  Thus, the transformed interval graph $I$ is such that $|V(I)|=O(n)$ and $m=|E(I)|$.  We shall now analyze the complexity of Algorithm \ref{hcalgo}. It is known from {Booth and leuker} \cite{golu} that maximal clique ordering (MCO) of interval graphs can be obtained in $O(n+m)$.   Using MCO, the height and level {information} can be obtained in $O(n+m)$ time.  Steps 4-12 of the algorithm essentially identifies the appropriate $y \in Y$ that can connect two consecutive $X$ vertices.   For a vertex $x_i \in X$, among unlabeled $Y$ vertices, the vertex whose degree is minimum is chosen.  Since the cost of this task is at most the degree of $x_i$ and the sum of the degrees is $O(m)$, the overall time complexity is $O(n+m)$ which is linear in the input size. \\
\begin{theorem}
Let $G$ be a convex bipartite graph satisfying Property A.  The Hamiltonian cycle in $G$ can be found in linear time.
\end{theorem}
\begin{proof}
Follows from Theorem \ref{hci2} and the above discussion. \qed
\end{proof}
We next present a necessary and sufficient condition for $G$ to have a Hamiltonian cycle, an important contribution of this section.
\begin{theorem} \label{hcyclechar}
Let $G$ be a convex bipartite graph. $G$ has a Hamiltonian cycle iff $G$ satisfies Property A.
\end{theorem}
\begin{proof}
\textbf{Necessary:} Suppose, $G$ has a Hamiltonian cycle $C$ and does not satisfy Property A.\\
\textbf{Case 1:} There exists a set $X_{p..q}$, such that $|N_G[X_{p..q}]| > q-p$, $1=p<q<n$ or $1<p<q=n$ or $1<p<q<n$. Then, 
Case 1.1: $C$ contains the set $X_{p..q}$ in order, and the set $Y$, alternately. Then, $C$ includes at most $q-p$ vertices that belong to  $N_G[X_{p..q}]$. Therefore, no Hamiltonian cycle exists in $G$. 
Case 1.2: $C$ contains the set $X_{p..q}$ out of order, and $Y$, alternately. Then, $C$ includes at most $q-p$ vertices that belong to  $N_G[X_{p..q}]$.  Therefore, $C$ is not a Hamiltonian cycle. \\
{
\textbf{Case 2:} There exists $j$ number of sets $X_{p_i..q_i}$, $1\leq i \leq j$ where $j\in [1..n-1]$ such that 
\begin{equation*}
	|\bigcup_{i=1}^j N'_G[X_{p_i..q_i}]|< (\sum_{i=1}^{j-1} q_i-p_i)+ q_j-p_j, 
	\begin{cases} & \text{if $j>1$, $ 1\leq p_i<q_i\leq p_{i+1}<q_{i+1}\leq n $}\\ & \text{if $j=1$, $1\leq p_j<q_j\leq n.$}\end{cases}
\end{equation*}
Let the set $Z=\bigcup_{i=1}^j N'_G[X_{p_i..q_i}]$ and $|Z|=(\sum_{i=1}^{j} q_i-p_i)-k$ where $k>0$. We observe that $C$ containing all the vertices in the set $Z$ can include at most $|Z|$ vertices from the set $\bigcup_{i=1}^j X_{p_i..q_i}$. However, $|\bigcup_{i=1}^j X_{p_i..q_i}|=(\sum_{i=1}^{j} q_i-p_i)+j>|Z|$, a contradiction. \\
\textbf{Case 3:} $|N_G[X_{1..n}]|\neq n$. We observe that $N_G[X_{1..n}]=Y=W \cup Z$, where $W = \{y~|~y \in Y$ and $d_G(y)= 1\}$, $Z = \{y~|~y \in Y$ and $d_G(y) \geq 2\}$ and $W \cap Z = \emptyset$. Since for a Hamiltonian cycle to exists, the graph must be 2-connected. This implies $|W|=0$. Thus, $N_G[X_{1..n}]=Y=Z$. For bipartite graphs to have a Hamiltonian cycle,  $|X|=|Y|$. This implies $C$ cannot exist with more than or less than $n$ vertices of $Y$.  \\
\textbf{Case 4:} $|N'_G[X_{1..n}]|\neq n$.  Clearly, no Hamiltonian cycle exists with less than or more than $n$ vertices from $Y$ since by definition, $|N'_G[X_{1..n}]|=Z=Y$, where $Z = \{y~|~y \in Y$ and $d_G(y) \geq 2\}$. Since we arrive at a contradiction in all cases, the necessary condition follows. \\
}
\textbf{Sufficiency:} Sufficiency follows from Theorem \ref{hci2}.   \qed
\end{proof}
{\bf Trace of Algorithm \ref{hcalgo}:}
We shall trace our algorithm with respect to the illustration given in Figure \ref{tracehcycle}.  A convex bipartite graph and its corresponding interval graph are shown.
\begin{figure}
\begin{center}
\includegraphics[scale=0.7]{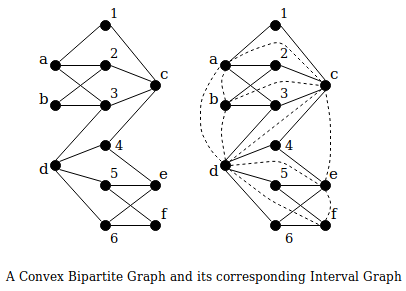}
\caption{An example to trace Hamiltonian cycle algorithm}\label{tracehcycle}
\end{center}
\end{figure}
\begin{itemize}
\item Computation of height and level {information}.  For the illustration, $C_1=\{1,a,c\}, C_2=\{2,a,c,b\}, C_3=\{3,a,b,c,d\}, C_4=\{4,c,d,e\}, C_5=\{5,d,e,f\}, C_6=\{6,d,e,f\}$.   Note that the height of $C_i$ is $i$.  Level {information is} as follows; $level[a]=3, level[b]=3, level[c]=4, level[d]=6, level[e]=6, level[f]=6$.
\item At iteration 1, i.e., $i=1$, the vertex $1$ is the base of the $U$-shaped path and the endpoints are $a$ and $c$.  Label $a$ as $A$ and $c$ as $B$.   Further, the vertex $1$ is labeled $O$.  The path constructed is $P(C_1)=(a,1,c)$.
\item $i=2$;  Using $P(C_1)$ and labels $A$ and $B$, we shall now extend the path. Since $level[A] \not= 2$ or $level[B] \not= 2$ and $level[A] > 2$ and $level [B] > 2$, we find $\min(level[A],level[B])=min(3,4)=3$.  This means, we extend $A$ by adding a $X$ vertex and a $Y$ vertex.  Since the unlabeled conductor of $C_2$ is $b$, we obtain $P(C_2)=(b,2,a,1,c)$.  Label $b$ as $A$ and label $a$ and {2} as $O$
\item $i=3$; Since $level[A]=i=3$, extend the path from $A$.  We obtain $(d,3,b,2,a,1,c)$.  Label $b$ and $3$ as $O$ and $d$ as $A$.
\item $i=4$; Since $level[B]=i=4$, extend the path from $B$.  We obtain $(d,3,b,2,a,1,c,4,e)$.  Label $c$ and {4} as $O$ and $e$ as $B$
\item $i=5$; Since $level[A]>i$ and $level[B]>i$, find $\min(6,6)$.  Choose either $A$ or $B$ to extend the path to get $(d,3,b,2,a,1,c,4,e,5,f)$.  Label $e$ and {5} as $O$ and label $f$ as $B$.
\item $i=6$; join $x_n$ to complete the cycle.  $(6,d,3,b,2,a,1,c,4,e,5,f,6)$
\end{itemize}
We next present an important result which says that {Chvatal's} necessary condition is indeed sufficient for 2-connected convex bipartite graphs. 
\begin{theorem}
Let $G$ be a 2-connected convex bipartite graph. $G$ has a Hamiltonian cycle if and only if for every $S \subseteq V(G)$, $c(G-S) \leq |S|$, where $c(G-S)$ denotes the number of connected components in $G-S$. 
\end{theorem}
\begin{proof}
We shall prove the sufficiency as the {proof} of necessary part follows from Chvatal's result.  The proof is by contradiction.  If, on the contrary, assume that $G$ does not have a Hamiltonian cycle.  By Theorem \ref{hcyclechar}, $G$ does not satisfy Property A.  The following case by case analysis completes the proof. \\
\textbf{Case 1:} $G$ violates the upper bound of Property A.  This implies that there exists a set $X_{p..q}$ for which 
$|N_G[X_{p..q}]| > q-p$, where $1=p<q<n$ or $1<p<q=n$ or $1<p<q<n$.  If we consider $S=X_{p..q}$, then $c(G-X_{p..q}) \geq |X_{p..q}|+1$, contradiction to the premise that $c(G-S) \leq |S|$.\\
\textbf{Case 2:} $G$ violates the lower bound of Property A.  There exists $j$ number of sets $X_{p_i..q_i}$, $1 \leq i \leq j$ where $j\in [1..n-1]$ such that 
\begin{equation*}
	|\bigcup_{i=1}^j N'_G[X_{p_i..q_i}]|< (\sum_{i=1}^{j-1} q_i-p_i)+ q_j-p_j, 
	\begin{cases} & \text{if $j>1$, $ 1\leq p_i<q_i\leq p_{i+1}<q_{i+1}\leq n $}\\ & \text{if $j=1$, $1\leq p_j<q_j\leq n.$}\end{cases}
\end{equation*} 
Then, $c(G-\bigcup_{i=1}^j N'_G[X_{p_i..q_i}])\geq (\sum_{i=1}^{j} q_i-p_i)+1>|\bigcup_{i=1}^j N'_G[X_{p_i..q_i}]|$.  A contradiction to the premise.\\
{\bf Case 3:}
By definition, it is clear that for $X_{1..n}$, $|N_G[X_{1..n}]| \geq |N'_G[X_{1..n}]|$.   Note that if $|N_G[X_{1..n}]| > |N'_G[X_{1..n}|$, then there exists a pendant vertex in $G$, however, we know that $G$ is 2-connected.  Therefore, $|N_G[X_{1..n}]|=|N'_G[X_{1..n}]|$.\\
\textbf{Case 3a:} $|N_G[X_{1..n}]|=|N'_G[X_{1..n}]|< n$. Then, $c(G-N_G[X_{1..n}])=n>|N_G[X_{1..n}]|$.\\
\textbf{Case 3b:} $|N_G[X_{1..n}]|=|N'_G[X_{1..n}]|> n$. Then, $c(G-X_{1..n})=|N_G[X_{1..n}]|>n$. \\
Since, we arrive at a contradiction in all of the above, our assumption that $G$ has no Hamiltonian cycle is false. \qed
\end{proof}
{\bf Remark:} If the above theorem is used for {the} algorithmic purpose, then test for Hamiltonicity requires exponential time in the input size. 
\section{Hamiltonian paths in convex bipartite graphs} \label{hpath}
In this section, we shall present some structural results and a polynomial-time algorithm to find Hamiltonian paths in {monotone} convex bipartite graphs.   
\begin{definition} Let $G$ be a convex bipartite graph. $G$ is said to satisfy {\tt Property B} if it satisfies the following properties 
\begin{enumerate}
\item Upper bound: \\
{For any set} $X_{p..q}$,\\
\begin{equation*}
  |N_G[X_{p..q}]|\leq
  \begin{cases}
    q-p+1, & \text{if $1=p<q<n$ or $1<p<q=n$ or $1<p<q<n$}.\\
    q-p+2, & \text{if $p=1$ and $q=n$}.
  \end{cases}
\end{equation*}
\item Lower bound: \\
For any integer $j\in [1..n-1]$,\\
\begin{equation*}
	|\bigcup_{i=1}^j N'_G[X_{p_i..q_i}]| \geq (\sum_{i=1}^{j-1} q_i-p_i)+ q_j-p_j, 
	\begin{cases} & \text{if $j>1$, $ 1\leq p_i<q_i \leq p_{i+1}<q_{i+1}\leq n$}\\ & \text{if $j=1$, $1\leq p_j<q_j\leq n$ }\end{cases}
\end{equation*}
\item Pendant bound: \\
The number of pendants is at most 2 and for $n \geq 2$, there is no $x \in X$ having two pendant $Y$ vertices incident on it.\\
\end{enumerate}
\end{definition}
To present algorithmic results, we partition convex bipartite graphs satisfying Property B into two sets, namely {\tt monotone} and {\tt non-monotone}, and their definitions are as follows;
\begin{definition} Let $G$ be a convex bipartite graph satisfying Property B.  $G$ is monotone if it satisfies the following properties;
\begin{enumerate}
    \item There does not exist a pendant vertex $y \in Y$ such that $y$ is adjacent to $x_k \in X$, $1<k<n$.
    \item There does not exist a {\em maximal} $X_{p..q}$, $1<p<q<n$.
\end{enumerate}
\end{definition}\label{monoprop}
\begin{definition}  Let $G$ be a convex bipartite graph satisfying Property B.  $G$ is non-monotone if it satisfies at least one of the following properties;  
\begin{enumerate}
    \item There exists a pendant vertex $y \in Y$ such that $y$ is adjacent to $x_k \in X$, $1<k<n$.
    \item There exists a {\em maximal} $X_{p..q}$, $1<p<q<n$.
\end{enumerate}
\end{definition}
We next consider the Hamiltonian path problem and show that convex bipartite graphs with Hamiltonian paths satisfy Property B.
\begin{algorithm}[H]
\textit{Input:} A Monotone convex bipartite graph $G$.\\
\textit{Output:} A Hamiltonian path $P$ in $G$.
\begin{algorithmic}[1]
\STATE {Initially,} all vertices are unmarked and set the path $P=\emptyset$.   The algorithm proceeds by visiting $|X|$ and $|Y|$ vertices alternately. When a vertex $u$ is visited, we mark $u$ and include in $P$.
\IF{$|Y|=|X|-1$} \STATE{\tt /* this case is referred to as $X-X$ path */ }

	\FOR{$i=1$ to $|X|-1$}
	
		\STATE Choose an unmarked vertex $y \in N_G(x_i)$ such that $right(y)$ is minimum.  Mark $x_i$ and $y$. Update $P=(P,x_i,y)$
	
	\ENDFOR
	\STATE Mark $x_n$ and update $P=(P,x_n)$.

\ELSIF{$|Y|=|X|$}  \STATE{\tt /* this case is referred to as $X-Y$ path or $Y-X$ path */ }
\STATE{In this {case,} we get $X-Y$ path or $Y-X$ path.  In the next for loop, we explore $X-Y$ path and if $P$ cannot be extended at some stage, then the subsequent steps {explore} $Y$-$X$ {path.}}
	\FOR{$i=1$ to $|X|$}
	
\STATE Choose an unmarked vertex $y \in N_G(x_i)$ such that $right(y)$ is minimum.  Mark $x_i$ and $y$. Update $P=(P,x_i,y)$
\STATE Break if $P$ cannot be extended or there exists an unmarked vertex $y \in Y$.  Skip the next step if {the} $X-Y$ path is found.
	\ENDFOR

\STATE Choose an unmarked vertex $y \in N_G(x_1)$ such that $right(y)$ is minimum.  Mark $y$.  Update $P=(P,y,x_1)$.  Find $X-X$ path from $x_1$ in the graph induced on $V(G) \setminus \{y\}$. Update $P$ by augmenting the $X-X$ path.


\ELSIF{$|Y|=|X|+1$}
\STATE{\tt /* this case is referred to as $Y-Y$ path */}
	\STATE Choose an unmarked vertex $y \in N_G(x_1)$ such that $right(y)$ is minimum.   Mark $y$.  Update $P=(P,y,x_1)$.  Find $X-Y$ path from $x_1$ in the graph induced on $V(G) \setminus \{y\}$. Update $P$ by augmenting {the} $X-Y$ path.


\ENDIF
\STATE Output $P$ as the Hamiltonian path in $G$.
\end{algorithmic}
\caption{Hamiltonian paths in monotone convex bipartite graphs}\label{monoalgo}
\end{algorithm}

\begin{lemma}
Let $G$ be a convex bipartite graph. If $G$ has a Hamiltonian path then $G$ satisfies Property B.
\end{lemma}
\begin{proof} For a contradiction, assume that $G$ has a Hamiltonian path $P$ and does not satisfy Property B.\\
\textbf{Case 1:} There exists $X_{p..q}$ such that $|N_G[X_{p..q}]| > q-p+1$, $1= p<q<n$ or $1<p<q=n$ or $1<p<q<n$. Then, \\
Case 1.1: $P$ contains vertices from $X_{p..q}$ and neither $x_p$ nor $x_q$ as endpoints as well as penultimate vertices.   Consider $(u_1,\ldots,u_{q-p+1})$ which is a permutation of $X_{p..q}$, and $P$ is such that $V(P)$ contains this permutation as an ordered subset.  That is, the vertices in the chosen permutation appears in order in $P$.  It is clear that $P$ contains vertices from $X_{p..q}$ and $Y$ alternately.  Therefore, $P$ has at most $q-p$ vertices from $Y$ ($N_G[X_{p..q}]$).  This implies that $P$ is not a Hamiltonian path, a contradiction. \\
Case 1.2: Same as Case 1.1 except that either $x_p$ or $x_q$ as {endpoints}  or penultimate vertices in $P$.  In this case,  $P$ includes at most $q-p+1$ vertices that belong to $N_G[X_{p..q}]$ and hence $P$ is not a Hamiltonian path.\\
Case 1.3: In this case, the vertices in the permutation are visited out of order.   That is, $P$ is such that $x_l \in X \setminus X_{p..q}$ appears in the subpath containing $X_{p..q}$.  If suppose, neither $x_p$ nor $x_q$ as endpoints as well as penultimate vertices, then $P$ includes at most $q-p-1$ vertices of $N_G[X_{p..q}]$.  On the other hand,  $x_p$ or $x_q$ or both can appear as endpoints or penultimate vertices in $P$, $p\not=1$ and $q\not=n$.  Accordingly, we see that $P$ includes at most $q-p$ or $q-p+1$ vertices that belong to $N_G[X_{p..q}]$ and hence $P$ is not a Hamiltonian path.\\
\textbf{Case 2:} There exists $X_{p..q}$ such that $|N_G[X_{p..q}]|>q-p+2$, $p=1$ and $q=n$. Since any path in $G$ alternates between $X$ and $Y$ vertices, any Hamiltonian path contains at most $n+1$ vertices from $Y$, a contradiction.  \\ 
\textbf{Case 3:} There exists $j$ number of sets $X_{p_i..q_i}$, $1 \leq i \leq j$ where $j \in [1..n-1]$ such that  
\begin{equation*}
	|\bigcup_{i=1}^j N'_G[X_{p_i..q_i}]|< (\sum_{i=1}^{j-1} q_i-p_i)+ q_j-p_j, 
	\begin{cases} & \text{if $j>1$, $ 1\leq p_i<q_i \leq p_{i+1}<q_{i+1}\leq n $}\\ & \text{if $j=1$, $1\leq p_j<q_j\leq n.$}\end{cases}
\end{equation*}
Let the set $Z=\bigcup_{i=1}^j N'_G[X_{p_i..q_i}]$ and $|Z|=(\sum_{i=1}^{j} q_i-p_i)-k$ where $k>0$. We observe that any path containing all the vertices in the set $Z$ can include at most $|Z|+1$ vertices that belong to the set $\bigcup_{i=1}^j X_{p_i..q_i}$. Moreover, $|\bigcup_{i=1}^j X_{p_i..q_i}|=(\sum_{i=1}^{j} q_i-p_i)+j>|Z|+1$. Therefore, $G$ does not have a Hamiltonian path. \\
\textbf{Case 4:} The number of pendants is at least 3.  Clearly, no Hamiltonian path can have all three pendant vertices.  \\
Thus, our claim that convex bipartite graphs with Hamiltonian paths satisfy Property B is true. \qed
\end{proof}
We shall next focus on monotone convex bipartite graphs and present a linear-time algorithm (Algorithm \ref{monoalgo}) for HAMILTONIAN PATH.
\subsection{Proof of correctness of Algorithm \ref{monoalgo}}
\begin{lemma}
\label{g'}
Let $G$ be a monotone convex bipartite graph.  Consider the graph $G'$ induced on $V(G')= \\
V(G) \setminus \{x_n,y,z\}$, where $y \in N_G(x_n)$ such that $left(y)$ is maximum and $z$ is a pendant adjacent to $x_n$ if $z$ exists. (or)
$V(G) \setminus \{x_n,y\}$, where $y \in N_G(x_n)$ such that $left(y)$ is maximum. \\
Then $G'$ is monotone.  
\end{lemma}
\begin{proof}   
We {first show} that $G'$ satisfies Property B. To prove the upper bound, we consider the following cases; for any integer $p$ and $q$ such that
\begin{enumerate}
\item $1 \leq p < q \leq n-2$. Then, $N_{G'}[X_{p..q}]=N_G[X_{p..q}]$.
\item $1 < p \leq n-2$ and $q=n-1$.   {\bf Case: $z$ does not exist}.  Suppose $N_G(y) \subseteq X_{p..n}$.  Observe that $N_{G'}[X_{p..n-1}]=N_G[X_{p..n}] \setminus \{y\}$.  This implies that $|N_{G'}[X_{p..n-1}]|= |N_G[X_{p..n}]|-1 \leq (n-p+1)-1\leq (n-1)-p+1$ which is the desired upper bound of Property B in $G'$.  Otherwise, $N_G(y) \not \subseteq X_{p..n}$.  Then $y$ is adjacent to $x_l, 1 \leq l \leq p-1$.  In this case, $N_{G'}[X_{p..n-1}]=N_G[X_{p..n}]$.  Thus $G'$ satisfies the upper bound. 
\textbf{Case: $z$ exists}.  Suppose $N_G(y) \subseteq X_{p..n}$.  Observe that $N_{G'}[X_{p..n-1}]=N_G[X_{p..n}] \setminus \{y,z\}$.  Further, $|N_{G'}[X_{p..n-1}]|= |N_G[X_{p..n}]|-2 \leq (n-p+1)-2 \leq (n-1)-p$ which satisfies the upper bound of Property B in $G'$.  Otherwise, $N_G(y) \not \subseteq X_{p..n}$.   In this case, we observe that since $G$ satisfies the upper bound, the graph $V(G) \setminus \{z\}$ also satisfies the upper bound.  Thus, $G'$ satisfies the upper bound of Property B.
\item $p=1$ and $q=n-1$. {\bf Case: $z$ does not exist}.  Note that $N_{G'}[X_{1..n-1}]=N_G[X_{1..n}] \setminus \{y\}$ and $|N_{G'}[X_{1..n-1}]|= |N_G[X_{1..n}]|-1 \leq (n+1)-1 = (n-1)-1+2$ which satisfies the upper bound of Property B in $G'$.
{\bf Case: $z$ exists}. Note that $N_{G'}[X_{1..n-1}]=N_G[X_{1..n}] \setminus \{y,z\}$.  This implies that $|N_{G'}[X_{1..n-1}]|= |N_G[X_{1..n}]|-2 \leq (n+1)-2 \leq (n-1)-1+1$, satisfying the upper bound of Property B in $G'$.
\end{enumerate}
We next prove that in $G'$, the lower bound of Property B is satisfied.  Note that by definition, $z$ is not considered while computing the lower bound.   If $d_G(y)=2$, then for the set $X_{p..q}$, $1 \leq p < q \leq n-1$, the lower bound of Property B in $G'$ is {the} same as in $G$.\\
Consider the case when $d_G(y) > 2$.  Let $N_G(y)=\{x_k,\ldots,x_n\}$.  
Then, for any $X_{p..q}$, $1 \leq p < q \leq k$, the lower bound in $G'$ is {the} same as in $G$.  For any integer $j' \in [0..k-1]$ and $j \in [1..n-k]$, let  $Z'$ denote $\bigcup_{i'=1}^{j'} N'_G[X_{p_{i'}..q_{i'}}]$, {where} if $j'=0$, $|Z'|=0$ or if $j'=1$, $1\leq p_{i'} < q_{i'}\leq k$ or if $j'>1$, $1\leq p_{i'} < q_{i'}\leq p_{i'+1} < q_{i'+1} \leq k$, $Z$ denote $\bigcup_{i=1}^j N'_G[X_{p_i..q_i}]$, {where} if $j=1$, $k \leq p_i < q_i \leq n-1$ or if $j>1$,  $k\leq p_{i} < q_{i}\leq p_{i+1} < q_{i+1} \leq n-1$, $W'$ denote $(\sum_{i'=1}^{j'} q_{i'}-p_{i'})$ and $W$ denote $(\sum_{i=1}^{j} q_i-p_i)$.  It is clear that $|Z' \cup Z| \geq W'+W$.  We now claim that  $|Z' \cup Z| \geq W'+W+1$.  \\
We observe that $|Z| \geq W+1$. If, on the contrary, $|Z| < W+1$. Note that the set $Z$ includes $y$ as well.  Observe that in $G$,  $|Z \cup N'_G[x_{n-1..n}]| < W+1$, contradicting the lower bound of Property B.   Therefore,  $|Z| \geq W+1$.  This implies $|Z' \cup Z| \geq W'+W+1$.
We see that $\bigcup_{i'=1}^{j'} N'_{G'}[X_{p_{i'}..q_{i'}}] \cup \bigcup_{i=1}^j N'_{G'}[X_{p_i..q_i}]= Z'\cup Z \setminus \{y\}$.  This means that $|\bigcup_{i'=1}^{j'} N'_{G'}[X_{p_{i'}..q_{i'}}] \cup \bigcup_{i=1}^j N'_{G'}[X_{p_i..q_i}]|=|Z'\cup Z|-1 \geq W'+W$. Therefore, $G'$ satisfies the lower bound of Property B.\\
Since in $G'$, no new vertex is introduced, the number of pendant vertices in $G'$ is still at most two.   In what follows from the argument is that $G'$ satisfies Property B, and further, it is easy to see that $G'$ satisfies the definition of monotone.  \qed
\end{proof}
\begin{lemma}\label{ob6}
Let $G$ be a convex bipartite graph satisfying Property B. If $G$ is monotone, then $|X|-1 \leq |Y| \leq |X|+1$.
\end{lemma}
\begin{proof}
It is easy to see that if $|Y| < |X|-1$, then $|N'_G[X_{1..n}]| \leq |Y| < n-1$, a contradiction to the lower bound of Property B.  Similarly, 
if $|Y|>|X|+1$, then $|N_G[X_{1..n}]| =|Y|> n+1$ which is a contradiction to the upper bound of Property B. \qed
\end{proof}
\begin{lemma} 
\label{l}
Let $G$ be a convex bipartite graph satisfying Property B.  If $G$ is monotone, then Algorithm \ref{monoalgo} gives a Hamiltonian path.
\end{lemma}
\begin{proof}
The proof is by induction on $|X|$.  \textbf{Base:} $|X|$=1. The possible instances are $|X|=1$ with $|Y| \in \{0,1,2\}$.  Clearly, the instances are paths of length $0$ or $1$ or $2$.  Thus, in all, the algorithm indeed gives the Hamiltonian path.  \textbf{Induction Hypothesis:} Let $G$ be a graph satisfying the premise and $|X| < n, n>1$.  Assume that the algorithm outputs a Hamiltonian path $P$ in $G$. \\
\textbf{Induction Step:} Let $G$ be a graph satisfying the premise such that $|X|= n, n >1$.  To use the induction hypothesis, we construct the graph $G'$ as defined in Lemma \ref{g'}.  Consider the graph $G'$ induced on $V(G')= \\
V(G) \setminus \{x_n,y,z\}$, where $y \in N_G(x_n)$ such that $left(y)$ is maximum and $z$ is a pendant adjacent to $x_n$. (or) \\
$V(G) \setminus \{x_n,y\}$, where $y \in N_G(x_n)$ such that $left(y)$ is maximum.  \\
{Further,} we know that $G'$ is monotone.  We make use of Lemma \ref{ob6} to complete the induction. \\
\textbf{Case 1:} $|Y|=|X|-1$.  We first observe that $z$ does not exist in $G$.  Suppose $z$ exists in $G$.  Then, $N'_G[X_{1..n}] \cup \{z\} \subseteq Y$.  Further, $|Y| \geq |N'_G[X_{1..n}]|+1=(n-1)+1=|X|$, a contradiction.   Clearly, in $G'$, $|X|=n-1<n$ and by the induction hypothesis, the algorithm computes $X-X$ Hamiltonian path $P$ in $G'$.  The path $P$ can be extended as $(P,y,x_n)$, the desired $X-X$ Hamiltonian path in $G$.\\
\textbf{Case 2:} $|Y|=|X|$.  Suppose $z$ exists in $G$.  Then, in $G'$, by the hypothesis, we get $X-X$  (since $|Y|=|X|-1$) Hamiltonian path $
P$.  Further, in $G$, $(P,y,x_n,z)$ is the desired $X-Y$ Hamiltonian path.  If $z$ does not exist, then by the induction hypothesis, we get either $X-Y$ path or $Y-X$ path $P$ in $G'$.  Accordingly, $(P,x_n,y)$ or $(P,y,x_n)$ is a Hamiltonian path in $G$. \\
\textbf{Case 3:} $|Y|=|X|+1$.  Suppose $z$ exists in $G$.  Then in $G'$, by the hypothesis, $Y-X$ path $P$ exists in $G'$.  The extended path $(P,y,x_n,z)$ is a Hamiltonian path in $G$.  Note that if $X-Y$ path exists in $G'$, then we claim that $Y-X$ path also exists in $G'$.  Suppose $Y-X$ path does not exist in $G'$.  Then as per Step 13 of our algorithm, while the algorithm is exploring $Y-X$ path, it breaks at  some iteration, say $l$.  This means at $x_l \in X$, there is no $y$ incident on $x_l$ using which the path can be extended further.  Clearly, the path has counted $l$ vertices in $X$ and $l$ vertices in $Y$.  The number of $Y$ vertices yet to be explored is $(n+1)-l$.   Note that when the algorithm breaks at $x_l$,  the unmarked vertices of $Y$ are such that their neighborhood is contained in $\{x_{l+1},\ldots,x_n\}$.  
Observe that $|N'_G[X_{1..l}]|=l$ and hence, $|N_G[X_{l+1..n}]|=(n+1)-l > (n-(l+1)+1)$.  This is in violation of the upper bound of Property B.  We further observe that the set $\{x_{l+1},\ldots,x_n\}$ is a maximal set.  However, we know that no maximal sets can exist in $G$ as $G$ is monotone.  This is a contradiction.   Thus, we always obtain a $Y-X$ path in $G'$ if exists.
If $z$ does not exist in $G$, then the path $P$ from the hypothesis together with $x_n$ and $y$ yields the desired Hamiltonian path in $G$.  \qed
\end{proof}
{\bf Run-time Analysis:} The primitive operation in our algorithm is to choose a $y$ vertex such that {the} degree of $y$ is minimum and this must be done at each $x \in X$.  Clearly, this task incurs $O(n+m)$ as the sum of the degrees is $O(|E(G)|)$.  Thus, our algorithm runs in linear time.
\begin{theorem}
Let $G$ be a convex bipartite graph satisfying the monotone property.  The Hamiltonian path problem is linear-time solvable.
\end{theorem}
\begin{proof}
Follows from Lemma \ref{l} and the above discussion.
\end{proof}
{\bf Trace of Algorithm \ref{monoalgo}:} An illustration given in Figure \ref{monofig} presents an example {of} monotone convex bipartite graphs for each category ($X-X$,$X-Y$, $Y-X$, $Y-Y$) along with the Hamiltonian paths traced as per our algorithm.\\\\
\begin{figure}
\begin{center}
\includegraphics[scale=0.7]{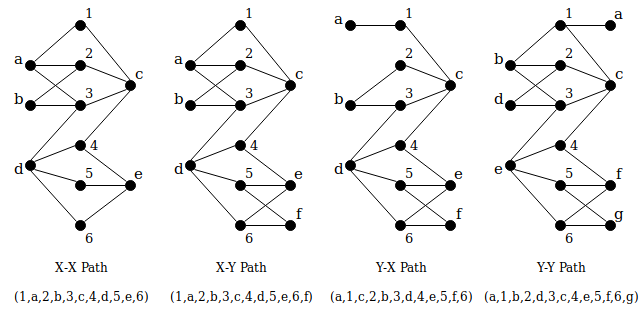}
\caption{Trace: Hamiltonian paths in monotone graphs}\label{monofig}
\end{center}
\end{figure}
\noindent
{\bf Remark:} In \cite{esha}, an algorithm to find {the} longest paths (and hence {the} Hamiltonian paths) in biconvex graphs is presented, however, the illustration given in Figure \ref{counteresha} shows that the algorithm does not output the solution correctly.   For the example given, the Hamiltonian path is $(d,3,b,1,a,2,c,4)$ whereas the output is $(1,a,2,b,3,c,4)$ which is not a Hamiltonian path.  We have given a sketch of how an infinite set of {counterexamples} can be constructed keeping the graph given on the left of the figure as a template.  Interestingly, all {counterexamples} are non-monotone convex bipartite graphs and hence, the complexity of {Hamiltonian} path in non-monotone convex bipartite graphs are open.  As far as monotone convex bipartite graphs are concerned, the algorithm in \cite{esha} gives the output in $O(n^6)$ time, while our algorithm outputs in linear time. 
\begin{figure}
\begin{center}
\includegraphics[scale=0.7]{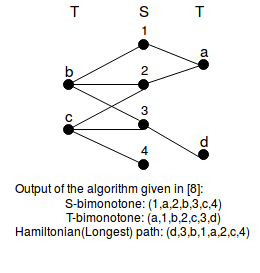}
\includegraphics[scale=0.5]{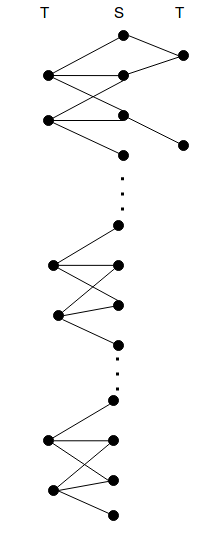}
\caption{A counterexample to Hamiltonian path algorithm of \cite{esha}} \label{counteresha}
\end{center}
\end{figure}
\section{Non-monotone Convex Bipartite Graphs}
\label{nonmono}
In this section, we shall present a few of our structural results which we believe can be used in discovering algorithms for Hamiltonian paths in non-monotone convex bipartite graphs.
\begin{lemma}
\label{c1}
Let $G$ be a convex bipartite graph satisfying Property B.  If $G$ has {\em maximal}$(X_{p..q})$ and {\em maximal}$(X_{p'..q'})$ such that $1 < p < q < n$, $1 < p' < q' < n$, $1 < p < p'$ and $q < q' < n$, then $p < q < p' < q'$.
\end{lemma}
\begin{proof}
If, on the contrary, $p < p' \leq q < q'$.\\ 
\textbf{Case 1:} $p < p' = q < q'$ and has no pendant vertex $y \in Y$ adjacent to $x_{p'}$. Then,
 $|N_G[X_{p..q'}]|=(q-p+1)+(q'-p'+1)=(q-p+1)+(q'-q+1)=q'-p+2$, violating the upper bound of Property B.\\
\textbf{Case 2:} $p < p' = q < q'$ and has a pendant vertex $y \in Y$ adjacent to $x_{p'}$. Then,
 $|N_G[X_{p..q'}]|=(q-p)+(q'-p')+1=q'-p+1$. This implies that the set $X_{p..q'}$ is maximal, violating the maximality of $X_{p..q}$ and $X_{p'..q'}$. \\
\textbf{Case 3:} $p < p' < q < q'$.  We now claim that $|N_G[X_{p'..q}]|=q-p'$.  We know that, if $|N_G[X_{p'..q}]|> q-p'+1$ then, the graph $G$ violates the upper bound of Property B. Now, we argue that neither $|N_G[X_{p'..q}]|=q-p'+1$ nor $|N_G[X_{p'..q}]|<q-p'$ holds true.  Assume on the contrary, \\
\textbf{Case 3a:} $|N_G[X_{p'..q}]|=q-p'+1$. Then, { we observe that $|N_G[X_{p..q}]\setminus N_G[X_{p'..q}]|=(q-p+1)-(q-p'+1)=p'-p$.  Similarly, $ |N_G[X_{p'..q'}] \setminus N_G[X_{p'..q}]|=(q'-p'+1)-(q-p'+1)=q'-q$.  Therefore, $|N_G[X_{p..q'}]|=|N_G[X_{p..q}] \setminus N_G[X_{p'..q}]|+|N_G[X_{p'..q'}] \setminus N_G[X_{p'..q}]|+|N_G[X_{p'..q}]|=(p'-p)+(q'-q)+(q-p'+1)=q'-p+1$. } This contradicts the maximal property of the sets  $X_{p..q}$ and $X_{p'..q'}$. \\
\textbf{Case 3b:} $|N_G[X_{p'..q}]| < q-p'$.   {Let $|N_G[X_{p'..q}]|=q-p'-k$, $k\geq 1$. Then, $|N_G[X_{p..q}]\setminus N_G[X_{p'..q}]|=(q-p+1)-(q-p'-k)= (p'-p+1+k)>(p'-p+1)$},  violating the upper bound of Property B.   Therefore, our claim $|N_G[X_{p'..q}]|=q-p'$ is true.  Thus, we get
\begin{align*}
|N_G[X_{p..q'}]| &= |N_G[X_{p..q}]|+|N_G[X_{p'..q'}]|-|N_G[X_{p'..q}]|\\ &=(q-p+1)+(q'-p'+1)-(q-p') \\ &=q'-p+2
\end{align*}
which is again the upper bound violation of Property B.  This completes the proof by contradiction and our claim that $p < q < p' < q'$ follows. \qed

\end{proof}
We define the set $Q=\{X_{p..q}~|~X_{p..q}$ is maximal, $1<p<q<n\}$  and the set $R=\{u~|~u \in Y$ and $u$ is pendant and there does not exist $X_{p..q} \in Q$ such that $u \in N_G(x_i), p \leq i \leq q \}$. 
In the next lemma, we show that when $G$ is non-monotone, we can bound the sum of $R$ and $Q$.
\begin{lemma}
Let $G$ be a convex bipartite graph satisfying Property B. If $G$ is non-monotone, then $1 \leq |R|+|Q| \leq 2$.
\end{lemma}
\begin{proof} We shall present a proof by contradiction.\\
\textbf{Case 1:} $|R|+|Q| = 0$. Follows from the definition that $G$ is monotone.\\
\textbf{Case 2:} $|R|+|Q| \geq 3$. \\
Case 2.1: $|R|\geq 3$ and $|Q|\geq 0$. Then, the pendant bound of Property B is violated.\\
Case 2.2: $|R|\geq 0$ and $|Q|\geq 3$.  Without loss of generality, let us assume $|Q|=k$, $k \geq 3$.  Note that $Q = \{X_{p_i..q_i} ~|~ X_{p_i..q_i}$ is {\em maximal}, $1 \leq i \leq k$ and $1 < p_i < q_i < n\}$.  Further, from Lemma \ref{c1}, we know that $p_i < q_i < p_{i+1} < q_{i+1}$.  From Property B, we know that,
\begin{align*}
|N'_G[X_{1..p_1}]\cup \bigcup_{j=1}^{k-1} N'_G[X_{q_j..p_{j+1}}]\cup N'_G[X_{q_k..n}]| &\geq (p_1-1)+ \sum_{j=1}^{k-1}(p_{j+1}-q_j)+(n-q_k)\\ &= (n-1)+(\sum_{j=1}^{k}p_j)-(\sum_{j=1}^{k}q_j) ~~~~~~~~~~~~~~~~~~~~~~~~~~~ (1)
\end{align*}
Note that by Property B, $|N_G[X_{1..n}]| \leq (n+1)$, and hence $|Y|\leq (n+1)$.  Since each $X_{p_i..q_i}$ is maximal, $|N_G[X_{p_i..q_i}]|=q_i-p_i+1$.  Thus, we get
\begin{align*}
|N'_G[X_{1..p_1}]\cup \bigcup_{j=1}^{k-1} N'_G[X_{q_j..p_{j+1}}]\cup N'_G[X_{q_k..n}]| &\leq  (n+1)-|(\bigcup_{j=1}^{k} N_G[X_{p_j..q_j}])|\\&= (n+1)-(\sum_{j=1}^{k}(q_j-p_j+1))\\ &= (n+1-k)+(\sum_{j=1}^{k}p_j)-(\sum_{j=1}^{k}q_j) ~~~~~~~~~~~~~~~~~~~~~~~~ (2)
\end{align*} 
Since $k \geq 3$, (2) contradicts (1).\\
Case 2.3: $|R|=1$ and $|Q|=2$. Let $X_{p_1..q_1}$ and $X_{p_2..q_2}$,  $1 < p_1 < q_1 < p_2 < q_2 < n$, are maximal in $Q$.  By Property B, we get,
\begin{align*}
|N'_G[X_{1..p_1}]\cup N'_G[X_{q_1..p_2}]\cup N'_G[X_{q_2..n}]| &\geq (p_1-1)+(p_2-q_1)+(n-q_2)\\ &= (n-1)+(p_1+p_2)-(q_1+q_2)
\end{align*}
Since $|Y|\leq (n+1)$ and $|N_G[X_{p_1..q_1}]|=q_1-p_1+1$ and $|N_G[X_{p_2..q_2}]|=q_2-p_2+1$, 
\begin{align*}
|N'_G[X_{1..p_1}]\cup N'_G[X_{q_1..p_2}]\cup N'_G[X_{q_2..n}]| &\leq (n+1)-|N_G[X_{p_1..q_1}]|-|N_G[X_{p_2..q_2}]|-|R|\\ &= (n+1)-(q_1-p_1+1)-(q_2-p_2+1)-1\\ &= (n-2)+(p_1+p_2)-(q_1+q_2)
\end{align*}
which violates the above equation.\\
Case 2.4: $|R|=2$ and $|Q|=1$.  Let $X_{p..q}$ be maximal in $Q$.  Observe that
\begin{align*}
|N'_G[X_{1..p}]\cup N'_G[X_{q..n}]| &\geq (p-1)+(n-q)\\ &= (n-1)+(p-q)
\end{align*}
Further,
\begin{align*}
|N'_G[X_{1..p}]\cup N'_G[X_{q..n}]| &\leq (n+1)-|N_G[X_{p..q}]|-|R|\\ &=(n+1)-(q-p+1)-2\\ &= (n-2)+(p-q).
\end{align*}
which violates the above equation.\\
Case 2.5: $|R|=2$ and $|Q|=2$. {Let $X_{p_1..q_1}$ and $X_{p_2..q_2}$,  $1 < p_1 < q_1 < p_2 < q_2 < n$, are maximal in $Q$}. Note that
\begin{align*}
|N'_G[X_{1..p_1}]\cup N'_G[X_{q_1..p_2}]\cup N'_G[X_{q_2..n}]| &\geq (p_1-1)+(p_2-q_1)+(n-q_2)\\ &= (n-1)+(p_1+p_2)-(q_1+q_2)
\end{align*}
Further,
\begin{align*}
|N'_G[X_{1..p_1}]\cup N'_G[X_{q_1..p_2}]\cup N'_G[X_{q_2..n}]| &\leq (n+1)-|N_G[X_{p_1..q_1}]|-|N_G[X_{p_2..q_2}]|-|R|\\ &= (n+1)-(q_1-p_1+1)-(q_2-p_2+1)-2\\ &= (n-3)+(p_1+p_2)-(q_1+q_2)
\end{align*}
which violates the above equation.  From the above case analysis, it follows that $1 \leq |R|+|Q| \leq 2$.  \qed
\end{proof}
{\bf Conclusions and directions for further research:}  In this paper, we have presented a structural characterization and a linear-time algorithm for HAMILTONIAN CYCLE.  Further, on the class of monotone convex bipartite graphs, we have shown that HAMILTONIAN PATH is linear-time solvable.  We have also made an attempt in extending these ideas to study the Hamiltonian path problem in non-monotone convex bipartite graphs and longest path and minimum leaf spanning tree problems in convex bipartite graphs.   An interesting direction for further research is to study the complexity of dominating sets and its variants such as total dominating set, total outer connected dominating set, etc., restricted to convex bipartite graphs.  \\ \\
{\bf Acknowledgements:} The authors wish to appreciate  and thank N Narayanan (IIT Madras), C Venkata Praveena (IIITDM Kancheepuram) and P Renjith (IIIT Kottayam) for sharing their insights on this problem.

\end{document}